\documentclass{article}

\usepackage[preprint, nonatbib]{neurips_2021}
\usepackage{amsthm}
\newtheorem{proposition}{Proposition}
\newtheorem{definition}{Definition}
\newcommand{\argmin}{\mathop{\mathrm{argmin}}}
\newcommand{\argmax}{\mathop{\mathrm{argmax}}}

\usepackage[utf8]{inputenc} 
\usepackage[T1]{fontenc}    
\usepackage{hyperref}       
\usepackage{url}            
\usepackage{booktabs}       
\usepackage{amsfonts}       
\usepackage{nicefrac}       
\usepackage{microtype}      
\usepackage{xcolor}         
\usepackage{amsthm, amsmath, amsfonts, mathtools, amssymb} 
\usepackage{sgame, tikz} 
\usetikzlibrary{trees, calc} 
\usepackage{subfig} 

\title{Improving Social Welfare While Preserving Autonomy via a Pareto Mediator}

\author{%
  Stephen McAleer \\
  Department of Computer Science\\
  University of California, Irvine\\
  \texttt{smcaleer@uci.edu}
   \And
   John Lanier \\
   Department of Computer Science\\
   University of California, Irvine\\
   \texttt{jblanier@uci.edu}
   \AND
   Michael Dennis \\
   Department of Computer Science \\
   University of California, Berkeley \\
   \texttt{michael\_dennis@berkeley.edu}
   \And
   Pierre Baldi \\
   Department of Computer Science \\
   University of California, Irvine \\
   \texttt{pfbaldi@uci.edu}
   \And
   Roy Fox \\
   Department of Computer Science \\
   University of California, Irvine \\
   \texttt{royf@uci.edu}
}

\begin{document}

\maketitle

\begin{abstract}
Machine learning algorithms often make decisions on behalf of agents with varied and sometimes conflicting interests. In domains where agents can choose to take their own action or delegate their action to a central mediator, an open question is how mediators should take actions on behalf of delegating agents. The main existing approach uses delegating agents to punish non-delegating agents in an attempt to get all agents to delegate, which tends to be costly for all. We introduce a Pareto Mediator which aims to improve outcomes for delegating agents without making any of them worse off. Our experiments in random normal form games, a restaurant recommendation game, and a reinforcement learning sequential social dilemma show that the Pareto Mediator greatly increases social welfare.  Also, even when the Pareto Mediator is based on an incorrect model of agent utility, performance gracefully degrades to the pre-intervention level, due to the individual autonomy preserved by the voluntary mediator.
\end{abstract}

\section{Introduction}
We often choose to delegate our decisions to algorithms. 
What should a decision-support system do when multiple people choose to delegate their actions to the same algorithm? 
When someone uses an app to get directions, accepts recommendations for who to follow on social media, or takes a recommended job on a gig app, they are allowing an algorithm to make a decision for them which not only affects their own commute, news feed, or income, but also affects the outcomes of others. 



In this paper we consider a middle ground between central planning and anarchy 
where a third-party mediator is able to act on behalf of any player who chooses to delegate control to the mediator.
The mediator simultaneously plays moves on behalf of all players who have chosen to delegate to the mediator. This notion is natural in settings where there already exist reliable  administrators  in  place. For example, markets that take place on managed platforms, or computer networks that communicate through routers. 

Existing literature has proposed mediators that can implement efficient mediated equilibria~\cite{kalai2010commitment, rozenfeld2007routing, tennenholtz2004program}, but that probably will not work in practice, because they rely on punishing agents who do not delegate which can be costly to all agents involved. For example, on a platform for making restaurant reservations, if a single person does not allow the platform to choose a restaurant for them, the platform would send everyone else to the same restaurant to punish them, harming everyone.   

We introduce a mediator that is much less risky for agents to delegate actions to. Instead of using delegating agents as punishers, our mediator tries to maximize their utility subject to a Pareto improvement constraint. Our Pareto Mediator works as follows: every agent submits an action in the original game and a binary value signifying if they want the mediator to act for them. After everyone submits their actions, the mediator finds a set of actions for the delegating agents in the original game that maximizes the delegating agents' social welfare, without making any delegating agent worse off. 

Table \ref{prisoners} illustrates the proposed mediator in the prisoner’s dilemma. In the mediated game on the right, the top left quadrant is the original game which corresponds to both players choosing not to delegate. The third and fourth row and column correspond to taking the first and second actions but also choosing to delegate. The outcome $(2, 2)$ associated with both players cooperating is now a strong Nash equilibrium (D$+$, D$+$), with higher social welfare than the original Nash equilibrium. 

\begin{table}[t] 
\centering
\begin{game}{2}{2}[][]
	    &  C      &  D     \\
	 C  &  $2, 2$ & $0, 3$  \\
	 D  &  $3, 0$ & $1, 1$\\
\end{game}
\hspace{.2\columnwidth}
\begin{game}{4}{4}[][]
	    &  C$-$      &  D$-$     &  C$+$    & D$+$ \\
	 C$-$  &  $2, 2$ & $0, 3$ & $2, 2$ & $0, 3$ \\
	 D$-$  &  $3, 0$ & $1, 1$ & $3, 0$ & $1, 1$ \\
	 C$+$  &  $2, 2$ & $0, 3$ & $2, 2$ & $0, 3$ \\
	 D$+$  &  $3, 0$ & $1, 1$ & $3, 0$ & $2, 2$ \\
\end{game}
\vspace{1em}
\caption{Prisoner's Dilemma with a Pareto Mediator. \textbf{Left:} the original prisoner's dilemma. The only Nash equilibrium is for both players to defect (D, D), to receive a reward of $(1,1)$. \textbf{Right:} the mediated game with a Pareto Mediator. Agents submit an action (C or D) and a bit ($+$ or $-$) indicating whether or not they want to delegate to the mediator. If up to one agent delegates the game is unchanged. If both agents delegate, the Pareto Mediator replaces each outcome with another outcomes that Pareto dominates it. In this case the Pareto Mediator replaces $(1,1)$ with $(2,2)$. This results in (D$+$, D$+$) becoming a strong Nash equilibrium in the mediated game, with higher utility for both players.}\label{prisoners}
\end{table}

The Pareto Mediator has several desirable theoretical properties. First, in two-player games, delegating an action always weakly dominates taking that action in the original game. Second, all pure Nash equilibria in the mediated game have at least as much total utility as any pure Nash equilibria in the original game. Third, the Pareto Mediator is much less risky than the Punishing Mediator in most games of interest. If an agent is uncertain whether all other agents will delegate, it can be very risky to delegate to the Punishing Mediator, because the agent could carry the burden of punishing others. Conversely, delegating to the Pareto Mediator tends to be much less risky because the group of delegators are all guaranteed to get higher utility than they would have if none of them delegated. 

We run simulations with online learning agents on normal form games and a restaurant reservation game. Empirically, we find that the Pareto Mediator achieves much higher social welfare than the Punishing Mediator. Moreover, a much higher percentage of agents choose to delegate actions to the Pareto Mediator than to the Punishing Mediator. Importantly, when the mediator does not have true rewards for the agents, the agents choose not to delegate actions, and the welfare is better than a central planner would have achieved. We also run experiments with reinforcement learning (RL) agents in sequential social dilemmas and find that the agents always choose to delegate their actions to the Pareto Mediator and get much higher reward than with the Punishing Mediator.  

\section{Preliminaries}

\subsection{Games and Mediators}
A normal form game $\Gamma$ has three elements: a finite set of players $P = \{1, ..., N\}$, a pure strategy (action) space $S_i$ for each player $i$, and payoff functions $u_i(s)$ for each profile $s = (s_1, ..., s_N)$ of strategies. All players other than some player $i$ are denoted by $-i$. We focus on finite games where $S = \bigtimes_i S_i$ is finite. A mixed strategy $\sigma_i$ is a probability distribution over actions. The space of player $i$'s mixed strategies is denoted as $\Sigma_i$, where $\sigma_i(s_i)$ is the probability player $i$ assigns to $s_i$. We denote the expected payoff to player $i$ of mixed profile $\sigma$ by $u_i(\sigma)$. 

\begin{definition}
A mixed-strategy profile $\sigma^*$ is a Nash equilibrium if, for all players $i$, $u_i(\sigma^*_i, \sigma^*_{-i}) \geq u_i(s_i, \sigma^*_{-i})$ for all $s_i \in S_i$. If $\sigma^* \in S$ then it is a pure Nash equilibrium.
\end{definition}

Mediators are third parties to whom players in a game can delegate actions. In this paper we consider mediators that can condition on the actions of non-delegating players. 

A mediated game $\Gamma(M)$ takes a game $\Gamma$ and adds an extra binary action that agents take along with their original action to indicate whether to delegate their action to the mediator $M$. So for a player $i$, their mediated strategy space is $S_i \times \{0,1\}$. A mediator $M$ is a function that takes in a joint mediated strategy profile $s_m = (s_i, d_i)_{i \in P}$ and outputs a joint strategy profile $s' = M(s_m)$ for the original game, thus setting $u_i(s_m) = u_i(M(s_m))$ in the mediated game. Mediators may only act on behalf of agents that choose to delegate, so that $s'_i = s_i$ if $d_i = 0$. If $d_i = 1$, the mediator is free to choose a new action $s'_i$ for that player. We denote the set of delegating agents by $D \subseteq P$. 


\begin{figure}
\centering
\includegraphics[width=.3\columnwidth]{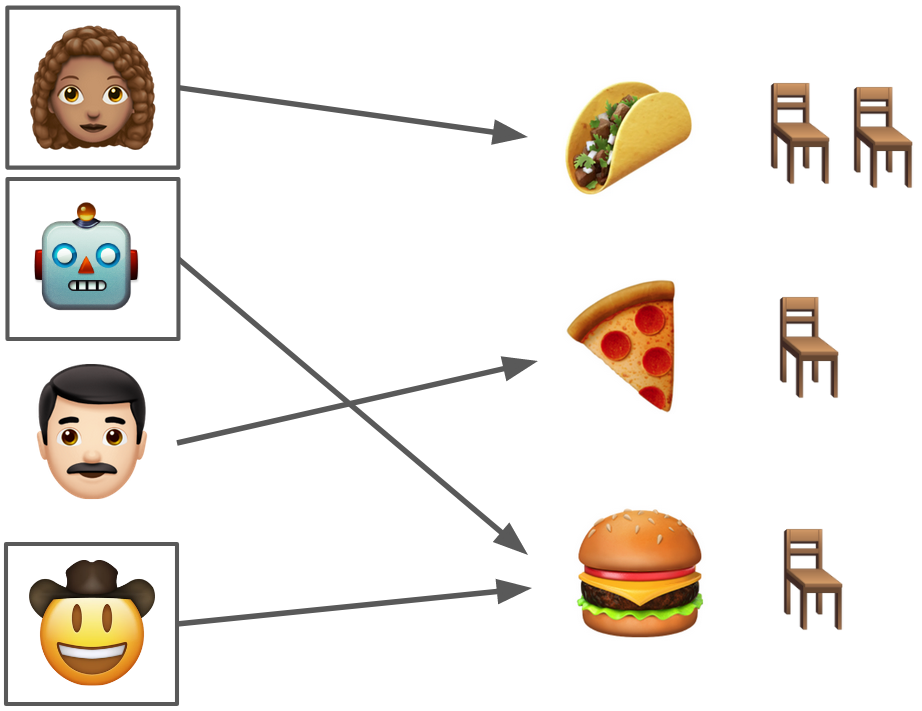}
\hspace{.03\columnwidth}
\includegraphics[width=.3\columnwidth]{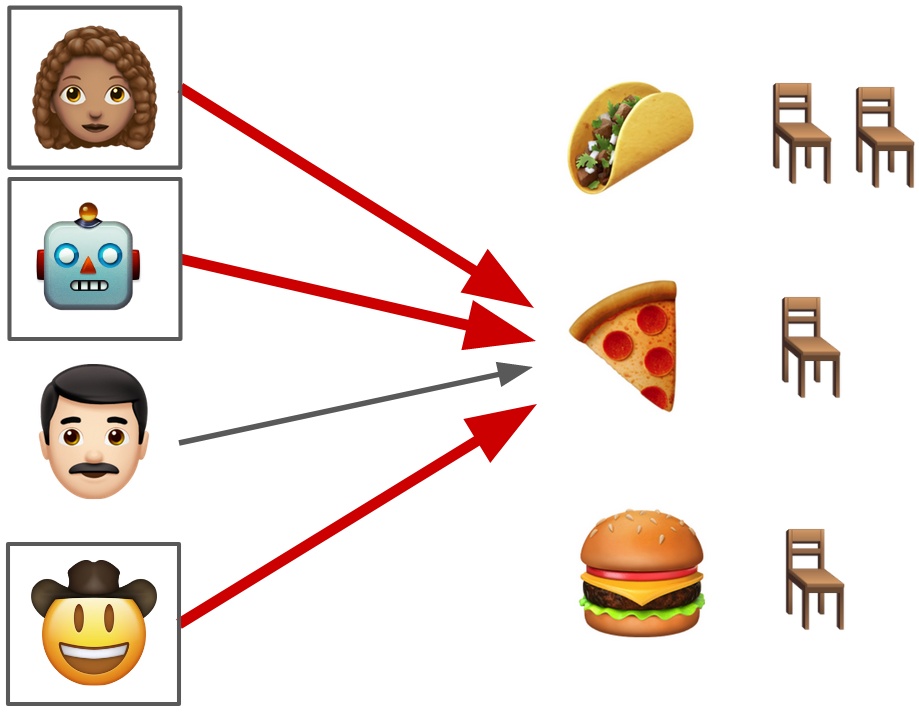}
\hspace{.03\columnwidth}
\includegraphics[width=.3\columnwidth]{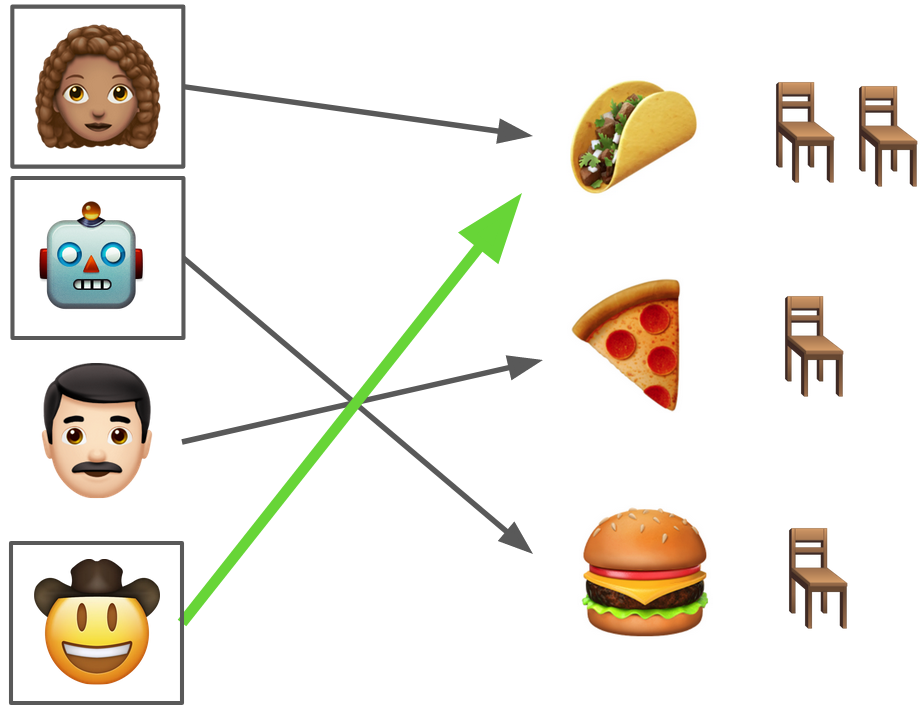}
\caption{\textbf{Left:} In the restaurant game (Section \ref{sec:rest}), agents simultaneously reserve restaurants (arrows) and can choose to delegate (square around their emoji). Restaurants have known capacity (number of chairs). If restaurant reservations exceed capacity, a random subset of agents get to eat there and the rest don't get to eat anywhere. \textbf{Center:} The Punishing Mediator makes all delegating agents go to the same restaurant as the non-delegating agent to maximally punish him, but at a high cost to the punishing agents as well. \textbf{Right:} The Pareto Mediator makes the cowboy go to a restaurant with extra capacity, which doesn't make any delegating agent worse off but makes the cowboy and the robot better off.}\label{fig:diagram}
\end{figure}

\subsection{Comparison to Other Concepts}
Mediators can be seen as a middle ground between anarchy, represented by the original game, and central planning, in which the best strategy profile is dictated. In our experiments and in previous theoretical results, it has been shown that allowing agents the ability to delegate to mediators can improve social welfare compared to social welfare under anarchy. Central planning can implement any outcome by forcing all agents to take particular actions, but in many situations such forcing may not be feasible or desirable. For example, when the mediator does not have full information about agent utilities, central planning may produce inferior solutions. Furthermore, central solutions concentrate power in the hands of one agent, which can make the system vulnerable to malicious attacks. 

Another type of mediator is one that is allowed to jointly recommend actions to all players. If no player has incentive to deviate from their recommended action, the joint distribution is called a \textit{correlated equilibrium} \cite{aumann1987correlated}. This mediator is weaker than the mediators we consider because it must allow for the possibility that its suggested actions are not subsequently taken by the players. On the other hand, correlated equilibria are more broadly applicable, due to the ubiquity of situations in which mediators are able to recommend actions but not take them on behalf of agents. However, as with the mediators we study, there remains an open question of how to choose which correlated equilibrium to suggest to the agents when there are several correlated equilibria.  

\subsection{Punishing Mediator}
One existing approach to implement desired outcomes by mediated equilibria is to use what we term a ``Punishing Mediator'' \cite{kalai2010commitment, rozenfeld2007routing, tennenholtz2004program}. If every agent in the game decides to delegate, then the Punishing Mediator will take actions for every agent to lead to the desired outcome. Otherwise, it will use the agents that did delegate to maximally punish the agents that did not delegate. The idea behind this mediator is to incentivize all agents to delegate, at which point the mediator implements an outcome that satisfies a social choice criterion.  

If not all agents delegate, the Punishing Mediator minimizes (perhaps approximately) the social welfare of non-delegating agents
\begin{equation}
    \begin{aligned}
    \argmin_{s'} \quad & \sum_{i \in P \setminus D} u_i(s')\\
    \textrm{s.t.} \quad & s'_j = s_j & \forall j \in P\setminus D.    \\
    \end{aligned}
    \label{punish}
\end{equation}

Otherwise if all agents delegate, the Punishing Mediator outputs the outcome that maximizes some social welfare function. In this paper we have the social welfare function be the sum of utilities 

\begin{equation}
    \begin{aligned}
    \argmax_{s'} \quad & \sum_{i \in P} u_i(s'),\\
    \end{aligned} \label{nice}
\end{equation}
but it could also be another social welfare function such as maximin.

Formally, the Punishing Mediator outputs the following joint strategy profile:
\[
    M(s_m) = \begin{cases}
        \text{Eq. }\eqref{punish} & \text{if } D \subsetneq P\\
        \text{Eq. }\eqref{nice} & \text{if } D = P.\\
        \end{cases}
\]


In order to incentivize every agent to delegate, the Punishing Mediator is very harsh when not every agent delegates. In games with a large number of agents that are learning and exploring, it can be unlikely for every agent to choose to delegate at the same time. In this setting, the Punishing Mediator lowers the utility of the non-delegating agents by having them punished and may also lower the utility of the delegating agents by having them punish the non-delegating agents. As we show in our experiments, this may cause all agents to have very low utility, and to opt out of delegating and punishing. In the next section, we introduce the Pareto Mediator which, instead of forcing delegating agents to punish non-delegating agents, tries to simply improve the delegating agents' utility. 


\section{Pareto Mediator}
The Pareto Mediator works by trying to make all the delegating agents better off without making any of them worse off. When the Pareto Mediator receives the joint strategy $s_m = (s_i, d_i)_{i \in P}$ of all players in the mediated game, it first calculates for each agent $i$ the utility $u_i(s)$ that it would have gotten in the original game. Next, the mediator finds mediated outcomes where every delegating agent $i \in D$ gets utility at least $u_i(s)$. Of these outcomes, it picks the one that maximizes some social welfare function for the delegating agents. As before, in this paper we use the sum of utilities as our social welfare function, but any could be used instead. 

Formally, the Pareto Mediator solves a constrained optimization problem to maximize social welfare for the delegating agents while ensuring that the utility they each get is not lower than they would have if nobody delegated. Then for a given game and joint mediated strategy $s_m$, if two or more agents delegate, the Pareto Mediator returns a (possibly approximate) solution to the following optimization problem:

\begin{equation}
    \begin{aligned}
    \argmax_{s'} \quad & \sum_{i \in D} u_i(s')\\
    \textrm{s.t.} \quad & u_i(s') \geq u_i(s) & \forall i \in D\\
    &s'_j = s_j & \forall j \in P\setminus D    \\
    \end{aligned}
\end{equation}


\begin{proposition}\label{prop:2pdom}
In two-player games, delegating to the Pareto Mediator is weakly dominant.
\end{proposition}
\begin{proof}
Let $s_m = (s_i, d_i)_{i=1,2}$ be a joint mediated strategy where agent $i$ delegates, $d_i = 1$, and $s' = M(s_m)$. Since the Pareto Mediator only changes agent actions if two or more agents delegate, if agent $-i$ does not delegate, $d_{-i} = 0$, then $M(s_m) = s$, and agent $i$ is indifferent between delegating or not, $u_i(s_m) = u_i(s)$. If agent $-i$ does delegate, then the mediator guarantees that agent $i$ is not worse off, $u_i(s_m) \geq u_i(s)$. 
\end{proof}
Proposition \ref{prop:2pdom} implies that in two player games, under mild conditions, players will always choose to delegate to a Pareto Mediator. Note that this property does not hold for the Punishing Mediator. However, in games with more than two players, not delegating may sometimes weakly dominate delegating to the Pareto Mediator. 

\begin{proposition}\label{prop:2ppure}
Any pure Nash equilibrium in a two-player game where both players delegate to the Pareto Mediator has at least as high total utility as any pure Nash equilibrium in the original game. 
\end{proposition}

\begin{proof}
The proof is contained in the appendix. 
\end{proof}

Consider a two-player game in which the agents learn through best-response dynamics, that is, they iteratively switch to their best strategy given the other player's strategy. This dynamics, if and when it converges, results in a pure Nash equilibrium. Since delegating weakly dominates not delegating in two-player games, Proposition \ref{prop:2pdom} suggests that in this equilibrium both players will often delegate, in which case Proposition \ref{prop:2ppure} guarantees that the social welfare will be at least as high as what the same dynamics could find in the original game. Also, since the outcomes only change if both agents delegate, Proposition \ref{prop:2ppure} implies that all pure Nash equilibria in the mediated game have higher total utility than the lowest utility pure Nash equilibrium in the original game. 



While both the Punishing Mediator and the Pareto Mediator have different static properties, we argue that much more important is their performance in non-equilibrium dynamics. Since the Punishing Mediator uses the delegating agents to punish the non-delegating agents, overall all agents will have a very low reward, unless all agents simultaneously decide to delegate. If that is not the case, overall utility can be much lower than even in the original game without a mediator. In our Pareto Mediator, instead of using delegating agents to punish non-delegating agents, the mediator tries to improve their social welfare. As a result, in non-equilibrium situations and in scenarios where not everyone is delegating, the Pareto Mediator will tend to have much higher social welfare than the Punishing Mediator. 

In the next section, we provide empirical evidence, through simulations on a number of different games, that (1) learning agents do not all delegate to the Punishing Mediator; and (2) the social welfare induced by the Pareto Mediator is higher than both the original game and the Punishing Mediator game.  

\begin{figure}[t]
\centering
\includegraphics[width=.3\columnwidth]{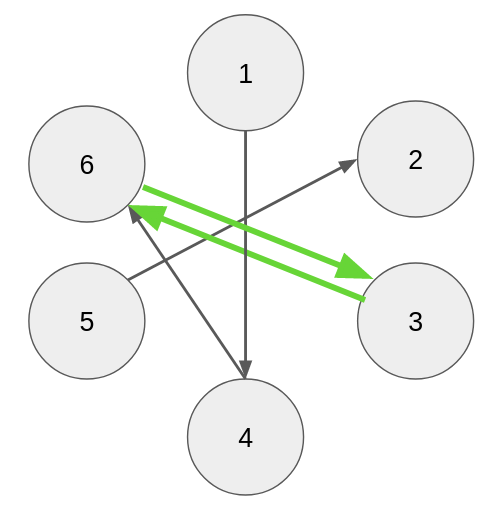}
\hspace{.1\columnwidth}
\includegraphics[width=.35\columnwidth]{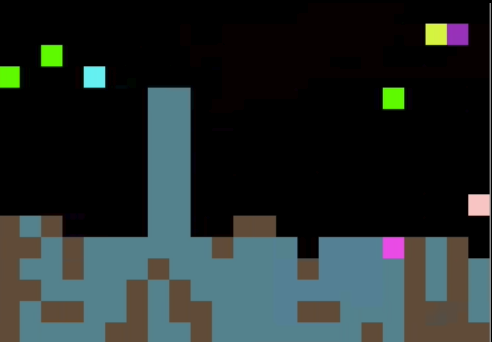}
\caption{In the matching game (left), agents simultaneously point to another agent. If both agents point to each other, they both get some positive reward, and otherwise get zero reward. In the cleanup environment (right), agents can choose to eat apples and get a positive reward, clean the river at a cost, or fire a beam to hurt the other player. If nobody cleans the river, apples will stop growing.}
\label{fig:games}
\end{figure}

\section{Experiments}
Agent-based simulations have long served as a tool in Game Theory for better understanding how agents dynamically learn in games \cite{davidsson2002agent, railsback2006agent}. While static analysis such as finding Nash equilibria is a useful tool, equally useful is the exercise of simulating how agents might play a particular game. In the experiments below, we use agents that are some form of reinforcement learning agents. The first three games are normal form games, and we use online learners, in particular epsilon-greedy agents where epsilon is decayed by $\frac{1}{t}$. In the third experiment, we test our approach on a sequential social dilemma, and use independent PPO \cite{schulman2017ppo} agents.  


\subsection{Random Normal Form Games}

\begin{figure}
\centering
\includegraphics[width=.3\columnwidth]{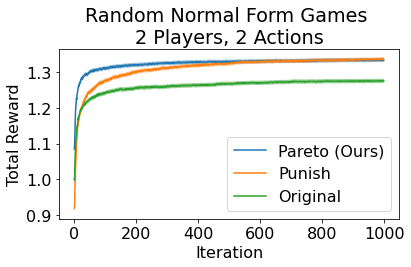}
\includegraphics[width=.3\columnwidth]{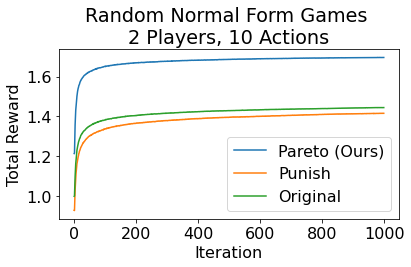}
\includegraphics[width=.3\columnwidth]{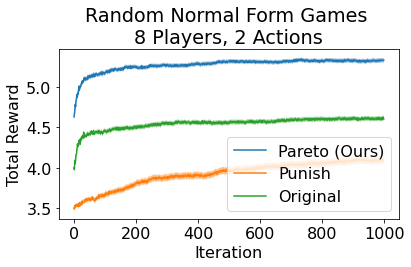}

\includegraphics[width=.3\columnwidth]{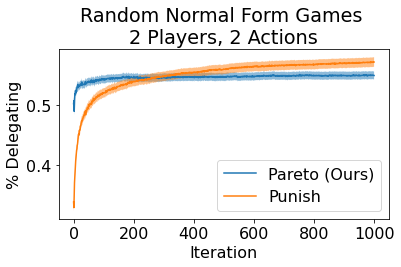}
\includegraphics[width=.3\columnwidth]{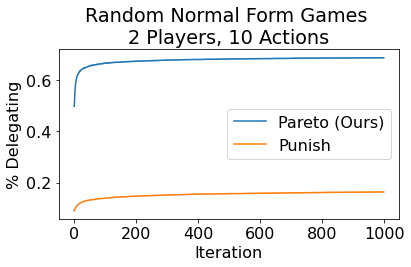}
\includegraphics[width=.3\columnwidth]{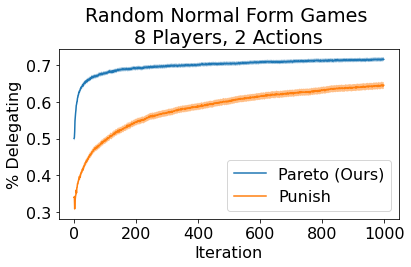}

\caption{\textbf{Top:} Average population reward with independent epsilon-greedy agents on random normal form games in the original as well as Pareto and Punishing mediated games. \textbf{Bottom:} Average percentage of game population agents who choose to delegate their actions to the mediators in corresponding Pareto and Punish games.}
\label{fig:normal_form_games}
\end{figure}

We create random normal form games by sampling outcome utilities independently from a continuous $[0,1)$ Uniform distribution. We then create new games that correspond to the mediated games from the Pareto Mediator and the Punishing Mediator. To create these games, we exactly solve the optimization problem by searching through all outcomes. Since these matrices grow exponentially in the number of players, we only include results on small games. As shown in Figure \ref{fig:normal_form_games}, we find that the Punishing Mediator can achieve total utility as high as the Pareto Mediator on random two-player, two-action games. However, as the number of actions or the number of agents increase, the performance of the Punishing Mediator becomes progressively worse. Interestingly, we also find that in games with larger action spaces, a very low percentage of agents choose to delegate actions to the Punishing Mediator. 

\subsection{Matching Game}

In the matching game (Figure \ref{fig:games}, left), every agent simultaneously picks a partner to match with. If two agents pick each other, they get a reward. Rewards are match-dependent and asymmetric. In these experiments, we sample rewards from a uniform distribution. The Pareto Mediator simply matches unmatched delegating agents with each other. The Punishing Mediator checks to see if any delegating agent is matched with any non-delegating agent and if so unmatches them. As shown in Figure \ref{fig:matching_game}, the Pareto Mediator gets much higher reward than the Punishing Mediator and the original game. We find that performance increases as the number of agents grows, intuitively because it becomes harder to find a match in the original game.  



\subsection{Restaurant Reservations}\label{sec:rest}
In the restaurant game, as shown in Figure \ref{fig:diagram}, agents can simultaneously make reservations at restaurants on a platform. Every restaurant has a different capacity level, and when it is exceeded each agent's utility decreases in proportion to the number of agents reserving at the same restaurant. Each agent has private utilities for each restaurant, but the platform only has access to public ratings. From these ratings, the platform builds a recommender system and uses these predicted ratings as utilities. In this experiment we use the Yelp Dataset \cite{yelp}, restricted to restaurants in Massachusets, and we only consider agents who have given 5 or more reviews. We then build a simple recommender system using collaborative filtering \cite{aggarwal2016recommender}.

\begin{figure}
\centering
\includegraphics[width=.3\columnwidth]{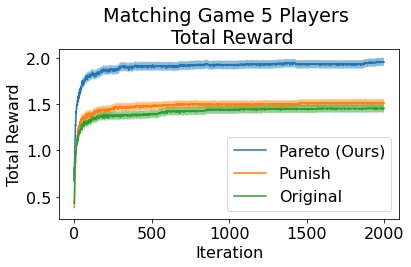}
\includegraphics[width=.3\columnwidth]{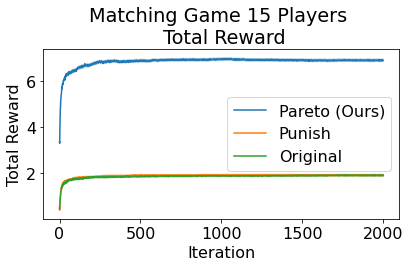}
\includegraphics[width=.3\columnwidth]{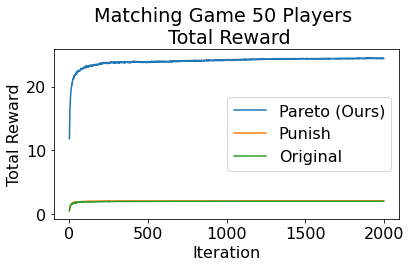}
\caption{In the matching game, as the number of agents increases, the Pareto Mediator achieves higher relative social welfare compared to the original game and the Punishing Mediator.}
\label{fig:matching_game}
\end{figure}

Formally, there are players $P$, and a set of restaurants $R$. For each $r \in R$, there is a restaurant capacity $c_r$. Every iteration, each agent $i$ takes an action $r_i$ which corresponds to picking a restaurant. Each agent has a restaurant-specific utility $u_i(r) = u^k_i(r) + \alpha u^u_i(r)$ where $u^k_i(r)$ is the known reward that is available to the platform via the recommender system and $u^u_i(r)$ is the agent-specific private reward, which we model as drawn from an independent Uniform $[0,1)$ distribution. The number of agents that choose a restaurant is given by $n_r = |\{i \in P : r_i = r\}|$. When $n_r \leq c_r$, each agent receives their restaurant-specific utility $u_i(s) = u_i(r_i)$. When $n_r > c_r$, each agent receives an expected value of snagging a table: $u_i(s) = \frac{c_{r_i}}{n_{r_i}}u_i(r_i)$. 

The Pareto Mediator works as follows. After every agent has submitted their actions, the non-delegating actions are taken as fixed. Next, new restaurant capacities are deduced and the Pareto Mediator approximately solves the max-utility objective with Pareto constraints by converting it into an assignment problem. To do so, the mediator creates a set of ``tables'' $T$, where each table $t$ is a unique remaining spot at a vacant restaurant $r(t)$. To account for the constraint that no delegating agent $i$ is worse off with a new choice $r'_i$ of restaurant instead of $r_i$, the Pareto Mediator creates a new utility function $\hat{u}$ whose value $\hat{u}_i(r')$ equals $u_i(r')$ if the restaurant is feasible ($u_i(r') \geq u_i(r_i)$), otherwise negative infinity. The Pareto Mediator then finds a bijection $f:D \mapsto T$ that solves the following assignment problem using a polynomial-time assignment problem solver \cite{crouse2016implementing}.
\begin{equation}\label{eq:assignment}
    \max_f \sum_{i \in D} \hat{u}_i(r(f(i))
\end{equation}
Note that this solution approximates the original social welfare objective, and we conjecture that improving it is sufficient for improving the overall performance of the Pareto Mediator. 

For the Punishing Mediator, if everyone delegates, we approximately solves the same assignment problem \eqref{eq:assignment}, but with $u$ instead of $\hat{u}$. If not everyone delegates, then the Punishing Mediator sends everyone to one restaurant that a non-delegating agent chose.  

We run two sets of experiments. In the first, we set $\alpha = 0$ so that the utility given by the recommendation system is the true agent reward. We include 300 restaurants and randomly sample restaurant capacities from $1$ to $10$.   
As shown in Figure \ref{fig:ssd} (left), the Pareto Mediator achieves higher total reward compared to the Punishing Mediator and the original game. The total utility achieved in this regime is nearly the same as the central planning solution, which is defined as the solution to \eqref{eq:assignment} with $D = P$ (all agents delegate). 
Furthermore, when we set $\alpha = 2$ so that the true reward is overshadowed by the unknown private reward (Figure \ref{fig:ssd}, center), total utility achieved by the Pareto Mediator is still higher than the original game and the Punishing Mediator. Notably, the Pareto Mediator reward is now much higher than the central planning solution. Since the central planner only has access to the recommendation system rewards, it often forces agents to go to restaurants they do not want to go to. By preserving autonomy through a voluntary mediator, higher social welfare can be achieved. 

\subsection{Sequential Social Dilemmas}
Finally, we test the effect of Pareto and Punishing mediators on independent reinforcement learners in the sequential social dilemma game Cleanup (Figure \ref{fig:games}, right) \cite{SSDOpenSource, leibo2017multi, Hughes2018, Jaques2019}. In this environment, two agents operate temporospatially in a grid-world containing apples and a river with waste in it. Waste spawns in the river periodically and inversely affects the frequency at which apples spawn in the world. Agents must balance the short-term gains from consuming apples for immediate reward with long term gains of cleaning the waist to increase the likelihood of more apples spawning. Agents are independently optimized with PPO \cite{schulman2017ppo}. They receive a top-down view of the world, can move in four directions, and must be near other objects in order to interact with them. Agents can also fire a beam at other nearby agents to inflict a reward penalty upon them. 

In the mediated versions of this game, during the first time step, agents can choose to delegate all further actions in the episode to the mediator, which will act on their behalf for the rest of the episode. The Pareto Mediator uses a joint policy that is trained to maximize cumulative reward. If one agent delegates but not the other, both agents play their original policies. But if both agents delegate, the mediator plays the joint pre-trained cooperative policy. The Punishing Mediator also plays the joint pre-trained cooperative policy if both agents delegate. But if only one delegates, the mediator plays a punishing policy that is pre-trained to minimize the opponent's reward.  

In Figure \ref{fig:ssd} (right), we show the average and minimum reward for both players over 100-episode windows throughout training. The Punishing Mediator achieves similar average utility as in the original game, but achieves much lower minimum reward than the original game. In contrast, when a Pareto Mediator option is provided, both the average and minimum population reward is much higher than in the original game and Punishing Mediator game. 
Agents playing the Pareto Mediator game learn to delegate 100\% of the time because doing so weakly dominated not delegating. However, agents playing the Punishing Mediator game never learn to delegate because most of the time when they delegate, the other agent was not delegating, and both agents received low reward by punishing or being punished.

\begin{figure}
\centering

\includegraphics[width=.3\columnwidth]{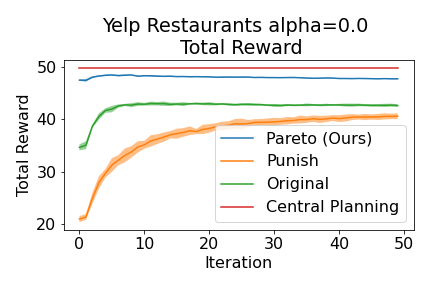}
\includegraphics[width=.3\columnwidth]{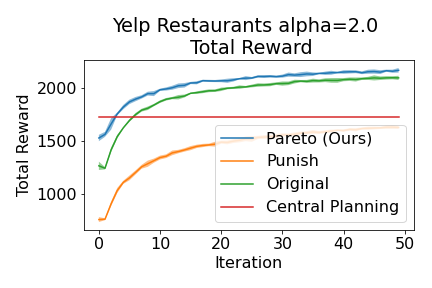}
\hspace{.05\columnwidth}
\includegraphics[width=.3\columnwidth]{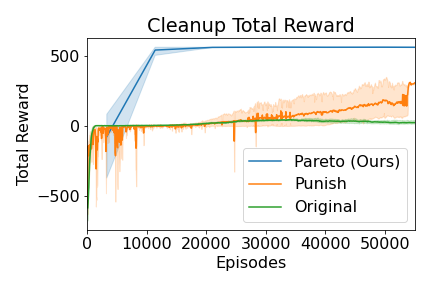}

\includegraphics[width=.3\columnwidth]{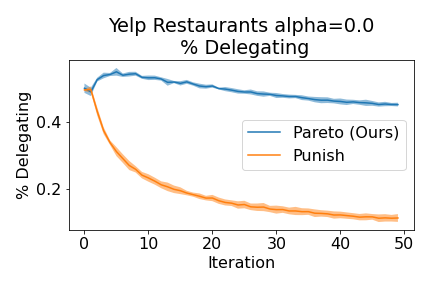}
\includegraphics[width=.3\columnwidth]{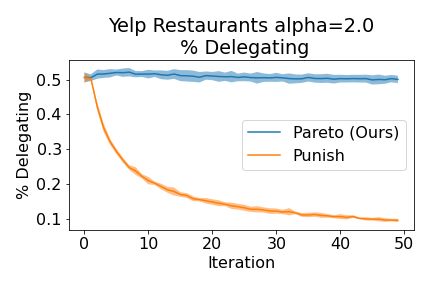}
\hspace{.05\columnwidth}
\includegraphics[width=.3\columnwidth]{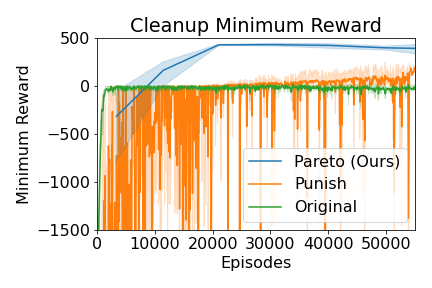}

\caption{\textbf{Left}: Total reward (top) and percentage of population delegating (bottom) in the restaurant game with $\alpha=0$ and $\alpha=2$. \textbf{Right}: Total reward (top) and minimum reward (bottom) in Cleanup.}
\label{fig:ssd}
\end{figure}

\section{Related Work}
\paragraph{Mediators and Commitment Devices.}
In economics there are many models for the design of mechanisms that maintain individual agent autonomy.  Program equalibria is a model where both players present a program to play on their behalf, which can then react to the program submitted by the other player, which allows for a kind of commitment device~\cite{tennenholtz2004program}.  More generally, commitment device games are a class of mechanisms which allow the designer to introduce an arbitrary class of commitment devices for the players to optionally access~\cite{kalai2010commitment}.  Our work builds off of a formally similar model of mediators, developed concurrently to commitment device games~\cite{tennenholtz2008game}.  There is a broad set of literature around mediators mostly focusing on analyzing their equalibria behavior in specific settings~\cite{peleg2010implementation,monderer2009strong,rozenfeld2007routing,peleg2007mediators,rozenfeld2006strong,ashlagi2009mediators}.  The most closely related is \cite{roth2020making} which proposes a general market matching algorithm for constructing a mediator.  Their construction ensures incentive compatibility of entering the market, and shows promising results for applications in labor markets.  However, though the aim of this work is similar, the mediator mechanism, its analysis, and the target domains are all quite distinct.

Correlated equilibria are a closely related concept to mediators in which agents are given the choice to follow the recommendations of a centralized recommender \cite{aumann1973subjecivity, foster1997calibrated}.  However, mediators are a stronger concept, as they allow the mediator not only to recommend actions, but to direct other agents to change their behavior based on which agents accept the recommendation, resulting in a strictly larger class of achievable equalibria.    

\paragraph{Mechanism Design and Solution Concepts.}
In the economics literature, designing incentives to motivate agents to behave in a desired fashion has long been studied in the field of mechanism design \cite{hurwicz2006designing}.  A popular objective in that literature is to design mechanisms which minimize the price of anarchy, which is the ratio of the worst Nash equilibrium to the maximal social welfare achievable by the mechanism \cite{koutsoupias1999worst,dubey1985inefficiency}.  The price of anarchy can be seen as a proxy for the true performance of the mechanism when deployed with learning agents.  An alternative model is to evaluate mechanisms based on how they empirically perform with popular learning models.  Learning in games is a broad field, and while many of these methods converge to popular equilibrium concepts in restricted settings~\cite{kalai1993rational,harada2006approximation,littlestone1994weighted}, many natural algorithms do not converge to Nash equilibria \cite{fudenberg1998theory}.  A construct that is often more predictive for these methods is risk dominance, which prioritizes one equilibrium over another if its utility is less sensitive to unpredictable behavior of other agents \cite{harsanyi1988general}.

Related to mechanism design is the field of social choice theory, which studies voting mechanisms which combine preferences to promote the general welfare~\cite{arrow2012social}.  There has been a large body of work proving impossibility results for social choice theory \cite{arrow2012social,myerson1983efficient,satterthwaite1975strategy,gibbard1973manipulation}. Though our work does not explicitly study social choice, it is closely related.  For instance, using a Pareto Mediator in the context of voting mechanisms would be a reasonable model for delegates or representatives.  Our approach avoids the impossibility results prevalent in social choice by having no notion of ``truthfully reporting your preferences''.

\paragraph{Multi-Agent Reinforcement Learning.}
Much work in multi-agent RL has focused on the setting of zero-sum games, to great success~\cite{heinrich2015fictitious,silver2017mastering,moravvcik2017deepstack,brown2017libratus, mcaleer2020pipeline, mcaleer2021xdo}.  However, in general-sum games, solutions are much less established.  In the context of this work, one line of work of particular relevance is the agenda around Cooperative AI~\cite{dafoe2020open}.  In broad strokes, this agenda lays out a set of problems which aim to use AI to further coordination and cooperation both in interactions between humans and AI and interactions between individual humans.  Our work can be seen as building in that direction by designing mechanisms to coordinate human actions while aiming to avoid reducing human autonomy.  One target domain of interest is that of studying RL systems in sequential social dilemmas~\cite{Perolat2017}.  This follows work on social dilemmas in game theory~\cite{Axelrod1981}, aiming to understand the effects of different strategies of promoting cooperation in sequential multi-agent environments~\cite{Foerster2018,Hughes2018,Eccles2019,Jaques2019,majumdar2020evolutionary,Letcher2019, roman2021accumulating}.  We propose our method as a mechanism for promoting cooperation without reducing the autonomy of the individual agents.

Our approach follows other work, using machine learning agents to better understand the effectiveness of certain mechanisms \cite{liu2020competing, johari2021matching}.  That work itself follows in a line of work aiming to better understand the multi-agent bandit setting more broadly \cite{cesa2016delay,shahrampour2017multi}.  We believe more theoretical insights from these directions could aid the analysis of the Pareto Mediator, though more work would be needed to achieve bounds that are applicable in our setting.

\section{Conclusion}
In this work we argued that mediators can serve as a middle ground between fully centralized control and anarchy. Existing mediators, such as the Punishing Mediator, 
achieve low social welfare when not at equilibrium. We propose the Pareto Mediator, which seeks to improve welfare of delegating agents without making any worse off. Through a series of experiments we show that the Pareto Mediator can be a simple and effective mediator with a wide range of potential applications.
Future work can extend Pareto Mediators to the transferable utility setting. In sequential domains, adding delegating options during the game instead of only at the beginning holds potential for preserving higher autonomy and addressing state-dependent uncertainty of agent utilities. Finally, the theoretical properties of the Pareto Mediator in static settings can be better understood, in order to provide guarantees and improved mechanisms. 



\bibliographystyle{abbrv}
\bibliography{mybib.bib}

\appendix
\section{Proof of Proposition 2}
\textbf{Proposition 2. }\textit{Any pure Nash equilibrium in a two-player game where both players delegate to the Pareto Mediator has at least as high total utility as any pure Nash equilibrium in the original game. }
\begin{proof}

Let $u_i(a, b)$ be the utility for player $i$ in the new game when player 0 takes action $a$ and player 1 takes action $b$. Similarly, $v_i(a, b)$ is the utility in the old game. To refer to an outcome for both players, we write $u(a,b)$ or $v(a,b)$. Define $U(a, b) := \sum_i{u_i(a, b)}$ and $V(a, b) := \sum_i{v_i(a, b)}$. If an outcome $u(a, b)$ Pareto dominates another outcome $u(a', b')$ then we write that as $u(a,b) \succeq u(a',b')$.  

Assume we have a pure Nash in the new game where player 0 takes action $a_0$ and player 1 takes action $a_1$ and that it has lower sum of utility than a pure Nash in the original game where player 0 takes action $b_0$ and player 1 takes action $b_1$. So $U(a_0,a_1) < V(b_0,b_1).$ 
\\
\\
Then it can't be the case that both $v(b_0,b_1) \succeq v(a_0,b_1)$ and $v(b_0,b_1) \succeq v(b_0,a_1)$. To see why, if it were the case, then $u(a_0,b_1) \succeq v(b_0,b_1)$ and $u(b_0,a_1) \succeq v(b_0,b_1)$ by definition of how we create the new game. But $u_0(a_0,a_1) \geq u_0(b_0,a_1)$ and $u_1(a_0,a_1) \geq u_1(a_0,b_1)$ because $u(a_0,a_1)$ is a Nash in the new game. So then $U(a_0,a_1) = u_0(a_0,a_1) + u_1(a_0,a_1) \geq u_0(b_0,a_1) + u_1(a_0,b_1) \geq v_0(b_0,b_1) + v_1(b_0,b_1) = V(b_0,b_1)$ which is a contradiction. 
\\
\\
Because one of the original "corners" isn't Pareto dominated by the original Nash, we can assume it is the bottom left one. Without loss of generality, assume $v(b_0,b_1) \not\succeq v(b_0,a_1)$. We know that $v_1(b_0,a_1) \leq v_1(b_0,b_1)$ because $v(b_0,b_1)$ is a Nash. So it must be that $v_0(b_0,a_1) > v_0(b_0,b_0)$. 
\\
\\
We know that $u_0(a_0,a_1) \geq u_0(b_0,a_1)$ and $u_1(a_0,a_1) \geq u_1(a_0,b_1)$ because $u(a_0,a_1)$ is a Nash in the new game. Suppose that $v_1(a_0,b_1) \geq v_1(b_0,b_1)$. Then $U(a_0,a_1) = u_0(a_0,a_1) + u_1(a_0,a_1) \geq u_0(b_0,a_1) + u_1(a_0,b_1) \geq v_0(b_0,a_1) + v_1(a_0,b_1) > v_0(b_0,b_1) + v_1(b_0,b_1) = V(b_0,b_1)$ which is a contradiction. So it must be that $v_1(a_0,b_1) < v_1(b_0,b_1)$. 
\\
\\
Now since $V(b_0,b_1)$ is a Nash in the original game, $v_0(a_0,b_1) \leq v_0(b_0,b_1)$. Then in the old game, $v(b_0,b_1) \succeq v(a_0,b_1)$, so in the new game $u(a_0,b_1) \succeq v(b_0,b_1)$. But this causes a contradiction because then $U(a_0,a_1) = u_0(a_0,a_1) + u_1(a_0,a_1) \geq u_0(b_0,a_1) + u_1(a_0,b_1) \geq v_0(b_0,a_1) + u_1(a_0,b_1) > v_0(b_0,b_1) + v_1(b_0,b_1) = V(b_0,b_1)$. 
\end{proof}

\section{Description of Experiments}
\subsection{Yelp Restaurant Data}
In this experiment we use the Yelp Dataset \cite{yelp}, restricted to restaurants in Massachusets, and we only consider agents who have given 5 or more reviews. We then build a simple recommender system using collaborative filtering \cite{aggarwal2016recommender}. All code and data will be released for the final version. 

\subsection{Sequential Social Dilemmas}

The Pareto and Punish Mediator PPO policies used in the Cleanup Sequential Social Dilemma game were pretrained with independent optimization for each agent to maximize the sum of agent rewards (for Pareto) or the negative of the other agent's reward (for Punish).

The multiagent PPO optimization process was built on top of the RLlib framework \cite{pmlr-v80-liang18b}. Each independently learning agent used a single experience gathering worker process with a single environment.

The CNN-LSTM network model was shared between the actor and critic. The model was comprised of one CNN layer with six 3x3 filters, stride 1 and relu activation followed by a two fully-connected layers size of 32 with tanh activations, then an LSTM with cell size 256, and finally two branching fully connected layers for discrete action logits and value function outputs.

Any hyperparameters not listed below were set to the default RLlib version 0.8.5 values.

\begin{table}[ht]
\centering
\begin{tabular}{ll}
discount factor & 0.99 \\
GAE $\lambda$ & 1.0 \\
entropy coeff & 0.001 \\
clip param & 0.3 \\
KL coeff & 0.2 \\
KL target & 0.01 \\
value function loss coeff & 0.0001 \\
learning rate & 0.0001 \\
optimizer & Adam \cite{DBLP:journals/corr/KingmaB14} \\
grad clip & 40 \\
train batch size & 1000 \\
sgd minibatch size & 250 \\
num sgd iters on train batch & 10 \\
LSTM unroll length & 20
\end{tabular}
\label{Tab:ppo-params}
\caption{PPO hyperparameters}
\end{table}

We used a single NVIDIA TITAN X GPU for all experiments and used 8 CPU cores. We will release all the code for all experiments in the final version.

\end{document}